\documentclass{article}
\usepackage[latin1]{inputenc}
\usepackage{graphicx}
\usepackage{color}
\usepackage{
	amssymb,
	amsmath,
	amsthm,
	wasysym,
	amssymb,
	mathrsfs,
	enumerate
}

\usepackage{relsize}
\newcommand\muchlarger[1]{%
\mathlarger{\mathlarger{\mathlarger{\mathlarger{#1}}}}
}
\newcommand\jonesBox[1]{%
  \ooalign{\raisebox{0.3pt}{$\muchlarger{\Box}$}\cr\;\raisebox{2.4pt}{\tiny$\tiny{#1}$}}
}
\newcommand\jonesDiamond[1]{%
  \ooalign{\raisebox{-0.7pt}{$\muchlarger{\Diamond}$}\cr\;\raisebox{2.4pt}{\tiny$\tiny{#1}$}}
}
\newcommand\plainDiamond{\mathlarger{\mathlarger{\Diamond}}}

\usepackage[english]{babel}
\usepackage{color}
\usepackage{hyperref}
\DeclareMathOperator{\pii}{pi}
\DeclareMathOperator{\va}{va}
\DeclareMathOperator{\vp}{vp}
\DeclareMathOperator{\Ought}{Ought}
\DeclareMathOperator{\viol}{viol}
\newcommand{\CJ}{CJ}
\newcommand{\SA}{SA}
\newcommand{\DD}{DD}
\newcommand{\FD}{FD}
\theoremstyle{plain}
\newtheorem{thm}{Theorem}
\newtheorem{cor}[thm]{Corollary}
\newtheorem{lem}[thm]{Lemma}
\newtheorem{rem}[thm]{Remark}
\newtheorem{res}[thm]{Result}

\newcommand{\norm}[1]{\|#1\|}

\begin{document}
	
	\setlength{\pdfpageheight}{\paperheight}
	\setlength{\pdfpagewidth}{\paperwidth}

	\title{Models of the Chisholm set}
	\author{Bj{\o}rn Kjos-Hanssen\footnote{
	This is a term paper for \emph{Filosofi hovedfag: spesialomr\aa de 1}, 1996, under the advisement of Andrew J.I.~Jones,
	translated from Norwegian.
	}}
	\date{November 11, 1996}
	\maketitle

	\begin{abstract}
		We give a counter-example showing that Carmo and Jones' condition 5(e) may conflict with the other conditions on the models in their paper
		\emph{A new approach to contrary-to-duty obligations}.
	\end{abstract}

	\tableofcontents
	\section{Deontic logic}
		Deontic logic (the logic of duties) has been studied more or less intensively
		throughout history. Given the emergence of modern modal logic one began to
		use logical systems to formalize valid logical deduction in connection
		with norms and obligations.

		We shall denote systems with capital letters like K, schemata with parenthesized capital letters like (K),
		and particular sentences with italicized letters like $A$.\footnote{
			This is consistent with Carmo and Jones' notation. It differs from Chellas' notation:
			he writes systems like $K$ in italics, and both schemata and particular sentences like A in roman.
		}

		The smallest normal modal logic system, K, is the smallest expansion
		of propositional logic with an operator ${\jonesBox{}}$
		(and a dual operator $\plainDiamond = \neg\jonesBox{}\neg$)
		having the axiom schema (K) and rule of inference (RN).
		\[
			\tag{K}\jonesBox{} (A \rightarrow B) \rightarrow (\jonesBox{} A \rightarrow \jonesBox{} B)
		\]
		\[
			\tag{RN}\vdash A \implies \vdash\jonesBox{} A
		\]
		Standard deontic logic KD is K expanded by the following axiom schema.
		\[
			\tag{D}\jonesBox{} A \rightarrow \plainDiamond A
		\]
		A standard model of modal logic is a structure $(W, R, \norm{\cdot})$ where $W$ is
		the set of possible worlds, $R$ is a binary relation on $W$, and
		$\norm{A} = \{\,x \in W {:} \models_x A\,\}$.
		A model of standard deontic logic is then
		a standard model of modal logic in which $R$ is serial,
		i.e., there is an ideal alternative $y$ for each world $x$,
		which in turn means that $\models_x\plainDiamond\top$,
		which in normal systems is equivalent
		to $\models_x \mathrm{D}$. This expresses that duties cannot be mutually contradictory.
		Indeed, if we employ an operator $\bigcirc$ for \emph{it is mandatory that}
		and assume that there is
		a proposition $A$ such that $\models_x \bigcirc A \land \bigcirc\neg A$,
		we obtain that in all ideal alternatives to $x$,
		$A\land\neg A$ holds, i.e., there is no ideal alternative to $x$.
		Since the symbol $\jonesBox{}$ is often reserved for concepts of necessity,
		one introduces a special symbol
		for deontic necessity, duty: $\bigcirc$, with the dual operator $\texttt{P}$ (permissibility).

		In a 1963 article \cite{Chisholm1963-CHICIA-2}, Chisholm pointed out that standard models are insufficient for
		deontic logic in the sense that they, as we shall see, necessarily lead to paradoxes.

		In this article we shall discuss Chisholm's paradox and a solution proposed by
		Carmo and Jones \cite{CJ96} in 1996. We shall call their system \CJ.

	\newpage
	\section{Contrary-To-Duty imperatives}
		Contrary-To-Duty (CTD) imperatives (or breach of duty imperatives) are duties
		that take effect when an ideal obligation has been violated.
		It was these that Chisholm showed can not be formalized in standard deontic logic.

		Chisholm's set of propositions is
		\begin{enumerate}[(a)]
			\item A certain man ought to go to help his neighbor.
			\item It ought to be the case that if he goes,
				he tells the neighbor that he is coming.
			\item If he does not go, he ought not to say that he is coming.
			\item He is not going.
		\end{enumerate}

		A formalization should take into account that intuitively this set is consistent
		and logically independent.
		If, as Chisholm, one formalized it in standard deontic logic
		KD, where the operator
		$\jonesBox{}$ 
		or $\bigcirc$ is understood as ``it is duty that'', as
		\[
			\bigcirc G\tag{a}
		\]
		\[
			\bigcirc(G  \rightarrow  S)\tag{b}
		\]
		\[
			\neg G \rightarrow \bigcirc\neg S \tag{c}
		\]
		\[
			\neg G\tag{d}
		\]
		then we can deduce
		$\bigcirc S$ (from (a) and (b)) and
		$\bigcirc\neg S$ (from (c) and (d)) which implies
		$\bigcirc\bot$ which via schema $D$ gives
		$\texttt{P}\bot$, which in turn yields that there is a world in which $\bot$,
		and we obtain a contradiction.
		Thus (a), (b), (c) and (d) cannot all hold. If we try using
		\[
			\tag{b'} ( G \rightarrow \bigcirc S)
		\]
		the set is no longer independent, as (b') follows from (d). If instead we put
		\[
			\tag{c'}
			\bigcirc (\neg G \rightarrow \neg S)
		\]
		we again do not retain independence, as (c') follows from (a).

		A reformulation which avoids the explicit reference to actions, and therefore
		often is preferable, is Prakken and Sergot's \cite{MR1400924} dog scenario:

		\begin{enumerate}[(a)]
			\item One ought not to have any dog.
			\item If one does have a dog, one ought to have a warning sign.
			\item One does have a dog.
		\end{enumerate}

		The fact that the Chisholm set cannot be formalized in standard deontic logic KD
		is traditionally called Chisholm's paradox. It is paradoxical only if one considers
		standard deontic logic as correct. A better name for it could thus be
		Chisholm's problem. Since the problem was posed by Chisholm in 1963,
		one of the greatest
		challenges for deontic logicians has been to find a solution: an alternative to
		standard deontic logic in which the Chisholm set can be formalized adequately.

		\subsection{The pragmatic oddity}
			In a precursor of \cite{CJ96},
			Jones and P\"{o}rn \cite{MR811630}] made a distinction between ideal and
			sub-ideal alternatives to a given world,
			using two relations $R$ and $R'$ in the semantics
			and corresponding normal strong operators $\jonesBox{i}$
			and $\jonesBox{s}$. One defines
			\[
				\jonesBox{\rightarrow} A   \iff \jonesBox{i} A  \land \jonesBox{s} A
			\]
			and
			\[
				\Ought A \iff \jonesBox{i} A \land \jonesDiamond{s} \neg A
			\]
			the latter in order to incorporate violability.
			One requires that the disjoint relations
			$R$ and $R'$ be serial, and that their union be reflexive.
			This implies that each world
			has ideal and sub-ideal alternatives, and is an ideal or sub-ideal alternative
			of itself.

			The formalization of the Chisholm set becomes
			\[
				\tag{a}  \Ought H 
			\]
			\[
				\tag{b}  \jonesBox{\rightarrow}  (H  \rightarrow  \Ought S)
			\]
			\[
				\tag{c}  \jonesBox{\rightarrow}  ( \neg  H \rightarrow  \Ought  \neg  S )
			\]
			\[
				\tag{d}  \neg  H 
			\]
			where
			\begin{itemize}
				\item $H =$ He goes to help his neighbor,
				\item $S =$ He lets the neighbor know that he is coming.
			\end{itemize}
			As was pointed out by Prakken and Sergot, we now get a subtle paradox which
			we may call the pragmatic oddity: We can deduce from (c) and (d) that
			$\Ought \neg  S$, and from (a) and (b) that
			$\jonesBox{i}\Ought S$. So in all ideal alternatives to the given world,
			he goes to help the neighbor, and has a duty to say that he is coming,
			but does not say that he is coming.

			We must be aware of the danger that the distinction between actual and ideal
			duties induces an infinite sequence of operators; one could end up with
			$\bigcirc_i= \bigcirc_1$ and then $\bigcirc_2, \bigcirc_3, \dots$ as a prioritized sequence of
			duties with varying degree of actuality. This problem arose in an article by
			Carmo and Jones \cite{CJ96}. In the dog scenario we may view the problem as follows: If
			there is a dog, no sign, and no fence, then what is the actual duty;
			to get rid of the dog, to put up a fence, or to put up a fence?

			The main idea in \cite{CJ96} is to incorporate
			the difference between the (f)actually possible and
			the potentially possible.
			It is an actual duty to put up a fence only if it is a factual
			necessity that there is a dog and no sign. Thus one avoids
			the pragmatic oddity and the infinite sequence of operators.
		\newpage
		\subsection{Requirements for any formalization}
			Seven reasonable requirements of any formalization of the dog scenario,
			put forward in \cite{CJ96}, are
			\begin{enumerate}
				\item Consistency and absence of moral conflict.
					The propositions in the dog scenario are
					consistent and do not contain any moral conflict;
					preferably one should not
					have a dog, and if that fails one should put up a sign;
					is this also violated then one ought to
					put up a fence.
				\item Logical independence. None of the propositions are redundant;
					the set is \emph{finite} in the sense that
					it is not equivalent to any proper subset of itself.
				\item The same logical form for contrary-to-duty conditionals
					as for other  deontic conditionals.
					Carmo and Jones point out that anything else makes a proposition's
					logical form depend on contingent facts (what happens to be duty)
					--- which is on a collision course with the existing paradigm of logic.
				\item Actual duties should be derivable.
					Depending on how inflexible the situation is,
					the actual duty should be to ensure either $\neg H$, $S$ or $G$.
				\item Ideal duties should be derivable.
					It is reasonable that it should be potentially possible for there
					to be no dog, so there should be an ideal duty to make sure
					that there is no dog.
				\item A duty has been violated, given that there is a dog.
					Thus we must have $\viol(\neg H)$.
				\item In no sense should the pragmatic oddity,
					that a sign ought to warn against a non-existent dog, appear.
			\end{enumerate}

			These are the requirements Carmo and Jones wanted to fulfill
			with the system we shall call \CJ.
	\newpage
	\section{The formal system \CJ}
		Let us write $\mathscr P$ for power set.
		\paragraph{Semantics.}
			$M =\langle W, \va, \vp, \pii, V \rangle$, where
			\begin{itemize}
				\item $W \ne \emptyset$,
				\item $V$ is a valuation
					(which we might just as well have called $\norm{\cdot}$),
				\item $\va : W \rightarrow \mathscr P(W)$,
				\item $\vp : W \rightarrow \mathscr P(W)$,
				\item $\pii : \mathscr P(W) \rightarrow \mathscr P(\mathscr P(W))$, and
				\item $w \in \va(w) \subseteq \vp(w)$.
			\end{itemize}
			(The real world is an actually possible world, and
			each actually possible world is a potentially possible world.)
			\begin{equation}
				\label{cj1}
				\emptyset \not\in \pii(X)
			\end{equation}
			(The contradictory context is not obligatory.)
			\begin{equation}
				\label{cj2}
				(Y \cap X = Z \cap X) \rightarrow
				(Y \in \pii(X) \leftrightarrow Z \in \pii(X))
			\end{equation}
			(Relevance: whether or not a context $Y$ is obligatory in a context $X$
			depends only on $X \cap Y$.)
			\begin{equation}
				\label{cj3}
				(Y \in \pii(X) \land Z \in \pii(X)) \rightarrow (Y \cap Z \in \pii(X))
			\end{equation}
			($\pii(X)$ is closed under intersection.)
			\begin{equation}
				\label{cj4}
				(X \subseteq Y \subseteq Z) \land (X \in \pii(Y )) \rightarrow
				(Z \setminus Y ) \cup X \in \pii(Z)
			\end{equation}
			With $\norm{B} \subseteq \norm{A} \subseteq W$ and
			\[
				A  \Longrightarrow_o B\quad
				=_{\text{df}}\quad
				\norm{B}  \in \pii( \norm{A} ),
			\]
			(\ref{cj4}) tells us that
			\[
				A  \Longrightarrow_o B
			\]
			implies
			\[
				\norm{A  \rightarrow  B}  \in \pii(W).
			\]
			It follows from conditions (\ref{cj1}) -- (\ref{cj4})) that
			\[
				X \in \pii(Y ) \land X \in \pii(Z) \Rightarrow X \in \pii(Y \cup Z).
			\]

		\paragraph{Truth conditions.}

			\begin{eqnarray*}
				\models_w \jonesBox{\rightarrow} A &\Longleftrightarrow&
				\va(w) \subseteq \norm{A}\\
				\models_w \jonesBox{} A &\Longleftrightarrow&
				\vp(w) \subseteq \norm{A}\\
				\models_w \bigcirc_a A &\Longleftrightarrow&
				\norm{A} \in \pii(\va(w)) \land \va(w) \cap \norm{\neg A} \ne \emptyset\\
				\models_w \bigcirc_i A &\Longleftrightarrow&
				\norm{A} \in \pii(\vp(w)) \land \vp(w) \cap \norm{\neg A} \ne \emptyset\\
				\models_w \bigcirc (B\mid A) &\Longleftrightarrow&
				\norm{B} \cap \norm{A} \ne \emptyset
				\land (\forall X \subseteq \norm{A})\\
				&&(X \cap \norm{B} \ne \emptyset \rightarrow \norm{B} \in \pii(X))\\
			\end{eqnarray*}

			We also define violation by $\viol(A)\Longleftrightarrow \bigcirc_i A \land \neg A$.

			\subsection{Proofs of results within the system}
				\begin{rem}[\texorpdfstring{\cite[Section 4.4, item 3]{CJ02}}{}]
					$\models \jonesDiamond{\rightarrow}$A $\rightarrow \plainDiamond$ A.
				\end{rem}
				\begin{proof}
					\[
						\va(w) \subseteq \vp(w),
					\]
					so
					\[
						(\exists w \in \va(x))(\models_w A) \implies (\exists w \in \vp(x))(\models_w A).\qedhere
					\]
				\end{proof}

				\begin{res}[Restricted factual detachment, \texorpdfstring{\cite[Section 4.4, item 14]{CJ02}}{}]\label{Resultat 1}
				\[
					\tag{a}
					\models
					\bigcirc(B\mid A)
					\land\jonesBox{\rightarrow}A
					\land\jonesDiamond{\rightarrow}B
					\land\jonesDiamond{\rightarrow}\neg B
					\rightarrow\bigcirc_a B
				\]
				\[
					\tag{b}
					\models
					\bigcirc(B\mid A)
					\land\jonesBox{}A
					\land\plainDiamond B
					\land\plainDiamond\neg B
					\rightarrow \bigcirc_iB
				\]
				\end{res}
				\begin{proof}
					(a) From
					\[
						\jonesDiamond{\rightarrow} \neg B
					\]
					we obtain
					\[
						\va(w) \cap \norm{\neg B} = \emptyset.
					\]
					From
					$\jonesBox{\rightarrow}$A and
					$\jonesDiamond{\rightarrow}$B we obtain

					\[
						\va(w) \subseteq \norm{A} \land\va(w)\cap \norm{B} = \emptyset
					\]
					and thus we obtain from $\bigcirc (B\mid A)$ that $\norm{B}\in \pii(\va(w))$,
					and we are done. (b) is similar.
				\end{proof}

				\begin{res}[\texorpdfstring{\cite[Section 4.4, item 15]{CJ02}}{}]\label{Resultat 2}
					\[
						\tag{a}
						\models
						\bigcirc(B\mid A)
						\land\jonesDiamond{\rightarrow}(A\land B)
						\land\jonesDiamond{\rightarrow}(A\land\neg B)
						\rightarrow\bigcirc_a(A\rightarrow B)
					\]
					\[
						\tag{b}
						\models
						\bigcirc(B\mid A)
						\land\plainDiamond(A\land B)
						\land\plainDiamond(A\land \neg B)
						\rightarrow \bigcirc_i(A \rightarrow B)
					\]
				\end{res}
				\begin{proof}
					(a) We use the following tricks: It follows from (2) in the semantics that
					\[
						\norm{A \rightarrow B} \cap\va(w) \in \pii(\va(w)) \iff \norm{A \rightarrow B} \in \pii(\va(w))
					\]
					since
					\[
						( \norm{A \rightarrow B} \cap\va(w))\cap\va(w) = ( \norm{A \rightarrow B} ) \cap \va(w).
					\]
					Thus it suffices to show that
					\[
						\norm{A \rightarrow B} \cap \va(w) \in\pii(\va(w)).
					\]
					Let $X = \norm{A} \cap \va(w)$. Since then $X \cap \norm{B} = \emptyset$ and $X \subseteq \norm{A}$, we obtain
					from $\bigcirc (B\mid A)$ that $\norm{B}\in\pii(X)$. Since
					$( \norm{B} ) \cap X = ( \norm{B} \cap X) \cap X$
					we have $\norm{B} \cap X \in \pii(X)$, i.e.,
					\[
						\norm{A} \cap \norm{B} \cap \va(w) \in \pii( \norm{A} \cap \va(w)).
					\]
					From (4) in the semantics we now get
					\[
						(\va(w)\setminus \norm{A} )\cup( \norm{A} \cap \norm{B} \cap\va(w))\cap\va(w) \in \pii(\va(w))
					\]
					which after some distribution of $\cap$ and $\cup$ gives
					\[
						\norm{\neg A \lor B} \cap \va(w) \in \pii(\va(w))
					\]
					and
					\[
						\norm{\neg A \lor B} \in \pii(\va(w)).
					\]
					(b) is similar.
				\end{proof}
				\begin{res}[Restricted deontic detachment \texorpdfstring{\cite[Result 2(vii) page 294]{CJ02}}{}]
					\[
						\tag{a}
						\models\bigcirc_aA\land\bigcirc(B\mid A)\land\jonesDiamond{\rightarrow}(A\land B)\rightarrow\bigcirc_a (A\land B)
					\]
					\[
						\tag{b}
						\models\bigcirc_iA\land\bigcirc(B\mid A)\land\plainDiamond (A \land B) \rightarrow \bigcirc_i (A \land B)
					\]
				\end{res}
				\begin{proof}
					Similar to the proof of Result \ref{Resultat 2}. The system does not have unlimited deontic
					detachment,
					\[
						\bigcirc_aA \land \bigcirc (B\mid A) \rightarrow \bigcirc_a B.\qedhere
					\]
				\end{proof}

				\begin{res}[Strong violability, \texorpdfstring{\cite[Section 4.4, item 12]{CJ02}}{}]
					\[
						\tag{a}
						 \models \jonesBox{\rightarrow} A \rightarrow ( \neg \bigcirc_aA \land \neg \bigcirc_a \neg A)
					\]
					\[
						\tag{b}
						 \models \jonesBox{} A \rightarrow ( \neg \bigcirc_iA \land \neg \bigcirc_i \neg A)
					\]
				\end{res}
				\begin{proof}
					(a) We have $\va(w) \cap \norm{\neg A} = \emptyset$ and thus $\neg \bigcirc_aA$. Furthermore, we have
					\[
						\norm{\neg A} \cap\va(w) = \emptyset \cap \va(w)
					\]
					and thus
					\[
						\norm{\neg A} \in \pii(\va(w)) \iff \emptyset \in \pii(\va(w)) \iff \bot.
					\]
					(b) is similar.
				\end{proof}

				\begin{res}[Strong classicality, \texorpdfstring{\cite[Section 4.4, item 13]{CJ02}}{}]
					\[
						\tag{a}
						 \models \jonesBox{\rightarrow} (A \leftrightarrow B) \rightarrow (\bigcirc_a A \leftrightarrow \bigcirc_a B)
					\]
					\[
						\tag{b}
						 \models \jonesBox{} (A \leftrightarrow B) \rightarrow (\bigcirc_i A \leftrightarrow \bigcirc_i B)
					\]
				\end{res}
				\begin{proof}
					(a) Since $\va(w) \cap \norm{A} = \va(w) \cap \norm{B}$ we obtain
					\[
						\norm{A} \in \pii(\va(w)) \iff \norm{B} \in\pii(\va(w))
					\]
					from the semantics. (b) is similar.
				\end{proof}

				\begin{cor}[Strong classicality II, \texorpdfstring{\cite[Result 2(vi) page 293]{CJ02}}{}]
					\[
						\tag{a}
						 \models \jonesBox{\rightarrow} A \rightarrow (\bigcirc_a B \rightarrow \bigcirc_a (A \land B))
					\]
					\[
						\tag{b}
						 \models \jonesBox{} A \rightarrow (\bigcirc_a B \rightarrow \bigcirc_a (A \land B))
					\]
				\end{cor}
				\begin{proof}
					(a) We have $\va(w) \subseteq \norm{A}$ and thus $\norm{A \land B} \cap \va(w) = \norm{B} \cap \va(w)$ which
					gives
					\[
						\norm{A \land B} \in \pii(\va(w)) \iff \norm{B} \in \pii(\va(w)).
					\]
					Moreover, we have
					\[
						\va(w) \cap \norm{\neg B} = \emptyset \Rightarrow \va(w) \cap \norm{\neg (A \land B)} = \emptyset.
					\]
					(b) is similar.
				\end{proof}

				The corollary may seem strange. If it is a factual necessity that there is a war
				in the world, and a factual duty to make sure that old ladies are helped crossing the street,
				does it follow that there is a factual duty to make sure that there is war and
				old ladies are helped across the street?

			\newpage
			\section{Analysis of \CJ\ and related systems}
				\subsection{The operators}
					In \cite{CJ96} the authors employ five operators:
					\begin{itemize}
						\item $\jonesBox{\rightarrow}$A (Weak necessity, \emph{It is not an actual possibility that not $A$})
						\item $\jonesBox{}$ A          (Strong necessity, \emph{It is not a potential possibility that not $A$})
						\item $\bigcirc_a$ A (Actual duty)
						\item $\bigcirc_i$ A (Ideal duty)
						\item $\bigcirc (B\mid A)$ (Conditional duty)
					\end{itemize}
					Note that the notation is inspired by probability theory, where $P(B\mid A)$
					often denotes the probability that $B$, given that $A$.
				\subsection{
					$(\forall x \in W)(\vp(x) = W)$
				}
					The way the system \CJ\ is defined, we may distinguish the following:
					\begin{itemize}
						\item The actually possible, $\models_x \jonesDiamond{\rightarrow} A$
						\item The potentially possible, $\models_x \plainDiamond A$
						\item That which is contingent but not potentially possible,
							($\models_x \jonesBox{} \neg A)\land(\exists y \in W)(\models y A$)
						\item The contradictory, $\models \neg A$
					\end{itemize}
					The effect of letting $\vp(x) = W$ is that the category for propositions that are not
					contradictory but still not potentially possible disappears. None of the scenarios
					discussed in \cite{CJ96} have any use for this category. Letting $\vp(x) = W$ is a
					simplification in the sense that we do not need to worry about what $\vp$ should be,
					and a anti-simplification in the sense that the definition of the semantics for the system
					becomes one line longer.

					It follows directly from the truth conditions that $\jonesBox{}$ and $\jonesBox{\rightarrow}$ are normal operators.
					Indeed, we can let
					\[
						xR_a y \iff y \in \va(x)
					\]
					and
					\[
						xR_i y \iff y \in \vp(x)
					\]
					where
					\[
						xR_a y (xR_i y)
					\]
					is read \emph{$y$ is an actual (potential) alternative to $x$}.

					The condition $w \in \va(w) \subseteq \vp(w)$ ensures that
					$\jonesBox{}$ and $\jonesBox{\rightarrow}$ satisfy the axiom T:
					\[
						( \jonesBox{}  A  \rightarrow  A)  \land  ( \jonesBox{\rightarrow} A  \rightarrow  A)\tag{T}
					\]
					$\bigcirc_a$ and $\bigcirc_i$
					are classical operators, and $\bigcirc (\cdot\mid\cdot)$ is classical with respect to each of its
					arguments (the antecedent and the consequent). This follows from the fact that its
					arguments, propositions $A$ and $B$, in the truth conditions only appear as $\norm{A}$ and $\norm{B}$.

					An operator $C$ is classical if it satisfies
					\[
						\norm{A} = \norm{B} \Rightarrow \vdash_x C(A) \iff C(B).
					\]
					For $\bigcirc_a$ and $\bigcirc_i$ we have
					violability ($\models \neg \bigcirc_{a,i}$) and
					fulfill-ability ($\models \neg \bigcirc_{a,i}\bot$)
					as well as closure under conjunction: the schema
					\[
						\tag{C} (\bigcirc_a A)  \land  (\bigcirc_a B)  \rightarrow  \bigcirc_a (A  \land  B).
					\]

					Violability follows from the truth conditions, whereas fulfill-ability and closure
					under conjunction are (1) and (3) in the semantics for $\pii$. The converse
					implication, the schema
					\[
						\tag{M} \bigcirc_a (A \land B) \rightarrow (\bigcirc_a A) \land (\bigcirc_a B),
					\]
					is in contradiction with violability;
					\[
						\bigcirc_a A \Rightarrow \bigcirc_a (A \land \top) \Rightarrow \bigcirc_a ( \top),
					\]
					so then, if there is an actual duty, all tautologies become obligatory.

				\subsection{$\bigcirc (B\mid A)$ is a classical conditional $A \Rightarrow B$}
					The truth condition for a classical conditional is
					\[
						\models_x A \Rightarrow B\quad \text{iff}\quad \norm{B} \in f(x, \norm{A} )
					\]
					for suitable $f$, whereas we in \CJ\ have
					\[
						\models_x A \Rightarrow \bigcirc  B
					\]
					(i.e., $\models_x \bigcirc (B\mid A))$ only if $\norm{B} \in \pii( \norm{A} )$,
					that is, independently of the possible world $x$. To replace \emph{only if} by \emph{iff} we define
					\begin{eqnarray*}
						f(x, X) = \{\,Y : (X \cap Y = \emptyset) \land\\
						(\forall Z)(Z \subseteq X \land Z \cap Y = \emptyset \rightarrow Y \in \pii(Z))\,\}
					\end{eqnarray*}
					which proves the proposition in the header: because then
					\[
						\models_x A \Rightarrow B
						\quad\text{iff}\quad
						\norm{B} \in f(x, \norm{A} )
						\quad\text{iff}\quad
						\norm{B} \in f(y, \norm{A} ).
					\]
					$\bigcirc (\cdot\mid\cdot)$ is still independent of $x$, as $x$ only
					appears formally as argument of $f$: we have $(\forall x, y \in W)(f(x, X) = f(y, X))$.

					Implication follows from equivalence, but closure under equivalence follows from closure under
					given that $\bigcirc $ commutes with $\land$, i.e.,that we have the schemata $M$ and
					$C$. Perhaps we may say in a sketchy way that closure under a weak condition is
					a stronger demand than closure under a strong condition.

				\subsection{Deontic and factual detachment}
					Deontic detachment is the schema
					\[
						\tag{\DD}
						\bigcirc (B\mid A) \rightarrow (\bigcirc A \rightarrow \bigcirc B)
					\]
					whereas factual detachment is
					\[
						\tag{\FD}
						\bigcirc (B\mid A) \rightarrow (A \rightarrow \bigcirc B).
					\]
					Deontic detachment suggests a reading of $\bigcirc (B\mid A)$ as
					\emph{it is obligatory that $B$, given that it is obligatory that $A$}, or
					\emph{ideally, it is obligatory that $B$ given that $A$}.

					Factual detachment suggests a reading of $\bigcirc (B\mid A)$ as \emph{it is obligatory that $B$,
					given that $A$ is the case}, or
					\emph{it is factually obligatory that $B$ given that $A$}.

					Carmo and Jones write that (FD) and (DD) allow for the deduction of, respectively,
					actual and ideal duties. The in contrast with their own system \CJ\ in which one
					in fact can derive both actual and ideal duties using limited versions of
					each of (FD) and (\DD). We may say that within deontic logic there is a
					(DD)-school, an (FD)-school and an (SA)-school, where (SA), strengthening of the antecedent,
					is the schema
					\[
						(A \Rightarrow \bigcirc B) \rightarrow (A \land A \Rightarrow \bigcirc B)
					\]
					i.e.,
					\[
						\bigcirc (B\mid A) \rightarrow \bigcirc (B\mid A \land A ).
					\]
					The schema (SA) is ill-fitting for, e.g., a conditional for typical circumstances,
					since the circumstances in which $A \land A'$ can be atypical considered as
					circumstances in which $A$. The modality \emph{under typical circumstances} has a parallel
					in deontic logic, \emph{under ideal circumstances}. If we allow both (FD) and
					(DD) we can deduce a moral conflict from the Chisholm set, and (SA) is
					incompatible with (FD):
					\[
						A \land C \land \bigcirc (B\mid A) \land \bigcirc (\neg  B\mid A \land C) \vdash_{\FD,\SA}
					\]
					\[
						\bigcirc B \land \bigcirc \neg B
					\]

					In \cite{CJ96}, limited versions of (FD) and (DD) hold, each in a version
					for factual and a version for ideal duties. (SA) holds under the condition that
					$\plainDiamond (A \land A' \land B)$.
					This follows directly from the truth conditions for $\bigcirc (\cdot\mid\cdot)$, since
					\[
						(\forall X)(X \subseteq \norm{A} \land X \cap \norm{B} = \emptyset \rightarrow \norm{B} \in \pii(X))
					\]
					implies
					\[
						(\forall X)(X \subseteq \norm{A \land B} \land X \cap \norm{B} = \emptyset \rightarrow \norm{B} \in \pii(X)).
					\]
					The logic \CJ\ thus does not take into account the problem associated with (SA). This
					means nothing more than: \CJ\ is a system with a limited purpose. When a duty
					is not fulfilled it could indeed be chalked up to two reasons: the duty only applies
					\emph{ceteris paribus}, or it has been violated. Only the latter possibility
					is treated in \cite{CJ96} and constitutes the core of Chisholm's problem.

				\subsection{Casta\~{n}eda and the distinction between practitions and propositions}
					Hector-Neri Casta\~{n}eda has made many contributions to deontic logic, among which one
					from \cite{MR645173} is considered here. The central idea is to distinguish
					between practitions and propositions. Practitions are expressions that grammatically
					tend to come after words like \emph{shall} and \emph{ought}. One lets the deontic operator
					be a function from the class of practitions to the class of propositions. That is,
					the proposition \emph{He ought to go help his neighbor} is analyzed as
					\[
						\Ought(\text{he}, \text{to go help the neighbor})
					\]
					and not as
					\[
						\Ought(\text{He goes to help the neighbor}).
					\]
					Casta\~{n}eda formalizes Chisholm's paradox as
					\begin{itemize}
						\item $\bigcirc A$
						\item $\bigcirc (a \rightarrow B)$
						\item $\neg  A \rightarrow \neg \bigcirc B$
						\item $\neg  A$
					\end{itemize}
					where practitions are written using capital letters ($A$), and propositions using lower-case
					letters ($a$).

					He uses the axiom
					\[
						\tag{KA} \bigcirc (A \rightarrow B) \rightarrow (\bigcirc A \rightarrow \bigcirc B)
					\]
					where $A$ and $B$ are practitions.

					We can deduce that $\neg  \bigcirc B$.
					We cannot deduce that $\bigcirc B$, since we do not have $\bigcirc (A \rightarrow B)$,
					but $\bigcirc (a \rightarrow B)$. We distinguish between the proposition \emph{he goes to help the neighbor}
					and the practition \emph{(he)(to go help the neighbor)}.

					Casta\~{n}eda admits that Chisholm's paradox can also be solved in other ways,
					but then only by creating unnecessarily complicated systems. Among these he also counts CJ.
					But he writes that Chisholm's paradox cannot arise in systems where one distinguishes between propositions and practitions.

					\paragraph{The objection to Casta\~{n}eda.}
						In the dog scenario it does not seem to be
						necessary to use practitions. If Casta\~{n}eda accepts the axiom
						\[
							\tag{Ka} \bigcirc (a \rightarrow b) \rightarrow (\bigcirc a \rightarrow \bigcirc b)
						\]
						he does not get rid of Chisholm's problem, and $Ka$ seems just as reasonable as $KA$.

				\subsection{The trajectory of philosophical logic}

					\begin{table*}[t]
						\centering
						\begin{tabular}{|c|c|c|}
							\hline
								& Monadic & Dyadic (C)\\
							\hline
							\hline
						Normal (K)
							& $\models \bigcirc (A \land B) \leftrightarrow \bigcirc A \land \bigcirc B$
							& $\models \bigcirc (A \land B\mid \top )\leftrightarrow  \bigcirc (A\mid \top ) \land \bigcirc (B\mid \top )$\\
						\hline
						Classical (E)
						    & $\bigcirc (B\mid A)$ not captured as & no problems\\
						    &  $A \rightarrow \bigcirc B$ nor as $\bigcirc (A \rightarrow B)$ & \\
						\hline
						\end{tabular}
						\caption{Problems with normal and monadic representation of deontic logic.}\label{tabell 2}
					\end{table*}

					The development within philosophical logic seems to be in the direction of greater complexity of syntax
					and semantics (see Table \ref{tabell 2}). Some argue for the use of quantifiers,
					temporal operators and action operators, while others like \cite{CJ96} take a more abstract approach
					while still finding a need for a rich conceptual apparatus.
					\CJ\ is also a fairly weak logic.
					The best label we can put on \CJ\ according to Chellas' classification \cite{MR556867}
					seems to be CECD': conditional logic (C) which is classical (E) and has
					the schemata (C) and (D'):
					\[
						\tag{D'} \neg\bigcirc\bot
					\]
					If one is to extrapolate from this one may end up with
					maxims like
					\begin{quote}
						The richer the semantics, the better.
					\end{quote}
					\begin{quote}
						The weaker the logic, the better.
					\end{quote}
					as a limiting case: For each semantics for deontic logic theere is then a richer
					semantics more suited for the purpose.

					In the transition from normal to classical systems, i.e., from standard models
					to minimal models, one obtains a richer semantics and a weaker logic. The same can be said
					about the transition from classical monadic systems to classical
					conditional logic systems like \CJ. The central deontic component in
					the semantics here, $\pii$, is a higher-level object in the following sense: where one in
					the original semantics of Kripke et al.\  only encounters semantic expressions
					of complexity $x \in \va(y)$, in \CJ\ one also finds, e.g., $X \in \pii(Y)$. Since
					here $X, Y \subseteq W = \norm{\top}$, i.e., $X, Y \in \mathscr P(W)$, we may say that we have applied
					the power set operation once to the semantics. See Table \ref{tabell 1}.

					\begin{table}
						\centering
						\begin{tabular}{|c|c|c|}
							\hline
							& Monadic & Dyadic\\
							\hline
							\hline
							Normal & $W \times W$ & $W \times W \times \mathscr P(W)$\\
							\hline
							Classical& $W \times \mathscr P(W)$& $W \times \mathscr P(W) \times \mathscr P(W)$\\
							\hline
						\end{tabular}
						\caption{Sets that the semantic function $f$ essentially is a subset of.}\label{tabell 1}
					\end{table}

					\paragraph{Lewis' contrafactual conditionals.}
						As with \CJ, Lewis'
						contrafactual conditional logic can be viewed as a classical conditional logic.

						Lewis' logic has an application as deontic conditional logic. Lewis' system
						also illustrates how semantics richer than Kripke semantics can be useful.
						In Lewis' system we have
						\[
						\models_x \text{If $A$ had been the case, then $B$ would have been the case}
						\]
						(a subjunctive conditional) interpreted as
						\begin{quote}
							There exist worlds where $A \land B$ that are closer to $x$ than all worlds in which $A\land \neg B$
						\end{quote}
						The deontic variant is
						\begin{quote}
							There exist worlds where $A \land B$ that are better,
							from the point of view of $x$, than all worlds in which $A\land \neg B$
						\end{quote}
						It is interesting how easy it is to incorporate ethical relativism in deontic
						logic: just let the truth value of a deontic proposition be dependent upon $x$.
						Propositions of the form $\bigcirc (B\mid A)$ in \CJ\ do not satisfy this in principle,
						but in practice the best fitting $X$ in the condition
						\[
							( \forall X)(X \subseteq \norm{A} \land X \cap \norm{B} = \emptyset \rightarrow \norm{B} \in \pii(X))
						\]
						will often be $\va(x)$ or $\vp(x)$. That is because we tend to have side conditions to $\bigcirc (B\mid A)$ of the
						type $\models_x \jonesBox{\rightarrow}$A and $\models_x \jonesDiamond{\rightarrow}$B, which gives exactly
						\[
							\va(x) \subseteq \norm{A} \land \va(x) \cap \norm{B} = \emptyset.
						\]
						The fact that Lewis' logic is a classical conditional logic we may obtain by letting
						\[
							f(x, X) = \{\,
								Y : ( \exists a \in X \cap Y )( \forall b \in X \setminus Y )
								(a \le_x b)
							\,\}
						\]
						where $a \le_x b$ means that $a$ is better than $b$ seen from $x$, morally speaking.
	\newpage
	\section{\CJ-models}
		\subsection[Counter-model for
			\texorpdfstring{$\bigcirc (B\mid A) \land \Diamond A \rightarrow \Diamond B$}%
			{$\bigcirc (B\mid A) \land \Diamond A \rightarrow \Diamond B$}
		]{
			Counter-model for $\bigcirc (B\mid A) \land \plainDiamond A \rightarrow \plainDiamond B$
		}
			Which \emph{Ought implies can} principles do we have in \CJ?
			Of course, we have
			\[
				\bigcirc_a A \rightarrow  \jonesDiamond{\rightarrow} A
			\]
			and
			\[
				\bigcirc_i A  \rightarrow \plainDiamond  A.
			\]
			But whether we should have, e.g.,
			$\bigcirc (B\mid A) \land \plainDiamond A \rightarrow \plainDiamond B$
			is a matter of interpretation.
			Here we shall show that this schema is not valid in \CJ,
			which is symptomatic for $\bigcirc (\cdot\mid\cdot)$'s independence of $\va$ and $\vp$.
			Generally, the diagnosis seems to have to be that in \CJ,
			\emph{Ought implies can} holds
			if \emph{can} is understood as logical possibility, but not
			if \emph{can} is understood as potential possibility.
			According to {\AA}qvist \cite{MR844606} there are two things one can
			ask in connection with Kant's \emph{Ought implies can} principle:
			what is meant by \emph{implies} and what is meant by \emph{can}.
			According to {\AA}qvist the most natural answer is that
			\emph{implies} is understood as logical consequence
			(the concept that philosophical logicians are striving to formalize) whereas
			\emph{can} is understood as a practical possibility,
			which in \CJ\ would mean $\plainDiamond$  or $\jonesDiamond{\rightarrow}$.

			\begin{thm}\label{counter-model}
				$\mathsf{CJ}\not\models \bigcirc (B\mid A) \land \plainDiamond A \rightarrow \plainDiamond B$.
			\end{thm}
			\begin{proof}
				To define our counter-model, let
				\begin{eqnarray*}
					W                               &=& \{w, y\}\\
					\norm{A}                        &=& W\\
					\norm{B}                        &=& \{y\}\\
					\vp(w) = \va(w)                 &=& \{w\}\\
					\vp(y) = \va(y)                 &=& \{y\} \text{ (say)}\\
					\pii( \emptyset ) = \pii(\{w\}) &=& \emptyset\\
					\pii(W) = \pii(\{y\})           &=& \{\{y\}, W\}\\
				\end{eqnarray*}
				Then we will get
				\[
					 \models_w \bigcirc (B\mid A) \land  \plainDiamond  A \land  \jonesBox{} \neg  B.
				\]
				The prove this we go through the semantic
				conditions for \CJ\ and check whether the model agrees with them all.

				The condition (\ref{cj2})
				\[
					\tag{\ref{cj2} revisited}
					(Y  \cap X  = Z  \cap  X)  \rightarrow  (Y  \in  \pii(X) \iff Z  \in  \pii(X))
				\]
				is trivially
				(in fact, in the sense of propositional calculus) satisfied for
				$X$ = $\emptyset$ and $X$ = \{w\} since
				$\models_{PC} p \rightarrow (\bot\iff\bot)$.)

				For $X$ = W we get just a set-theoretical fact,
				\[
					Y = Z  \rightarrow  (Y  \in  \pii(W)  \iff Z  \in  \pii(W))
				\]
				which entails no new information.

				For $X = \{y\}$ we must try $Y = \{y\}$ and then we get
				\[
					\{y\}  \cap  \{y\} = Z  \cap  \{y\}  \rightarrow  (\{y\}  \in 
					\pii(\{y\}) \iff Z  \in  \pii(\{y\}))
				\]
				so if $Z = W$ we get $W \in \pii(\{y\})$ which we already know.

				The condition (\ref{cj3})
				\[
				 	\tag{\ref{cj3} revisited}
					(Y  \in  \pii(X)  \land  Z  \in  \pii(X))  \rightarrow
					(Y  \cap  Z  \in  \pii(X))
				\]
				says that $\pii(X)$
				is closed under intersection, and this we can see is satisfied since
				$\{y\} \cap W = \{y\}$.

				For the condition (\ref{cj4})
				\[
					\tag{\ref{cj4} revisited}
					(X  \subseteq  Y  \subseteq  Z)  \land  (X  \in  \pii(Y ))  \rightarrow
					(Z  \setminus  Y )  \cup X \in  \pii(Z),
				\]
				we must choose $Y$ such that $\pii(Y ) = \emptyset$.

				Case (i): $Y = W$. Then we must have $Z = W$ and we just get
				\[
					 X \in  \pii(W)  \rightarrow X \in  \pii(W).
				\]
				Case (ii): $Y = \{y\}$.
				In order that $X \in \pii(Y )$ and $X \subseteq Y$ we must have $X = Y$.
				With $Z = X$ we then get a tautology whereas
				with $Z = W$ we get $W \in \pii(W)$ which we already have.

				Thus all the semantic conditions are satisfied,
				and the Theorem has been proved.
			\end{proof}
			If we also include the condition
			\begin{equation}
				\label{cj5}
				(Z  \in  \pii(X)) \land
				(Y  \subseteq  X) \land
				(Y  \cap Z \ne \emptyset )  \rightarrow  (Z  \in  \pii(Y ))
			\end{equation}
			then $\pii(\{w\}) = U(\{w\})$, where $U(X) = \{\,Y : X \subseteq Y \,\}$.

		\newpage
		\subsection{A three-point model $C_3$ for the Chisholm set}\label{three-point}
			Consider the following sentences $A$ and $B$.
			\[
				\text{He goes to help his neighbor.}\tag{$A$}
			\]
			\[
				\text{He says that he is coming.}\tag{$B$}
			\]
			A reasonable interpretation of the Chisholm set, mention in \cite{CJ96},
			suggests that the following propositions hold:
			\begin{align*}
				\bigcirc (A\mid \top )\\
				\bigcirc (B\mid A)\\
				\bigcirc (\neg  B\mid \neg  A)\\
				\jonesBox{\rightarrow}\neg  A\\
				\neg  B\\
				\jonesDiamond{\rightarrow} B\\
				\plainDiamond (A \land B)
			\end{align*}
			Thus, he has decided to not go to help,
			and so far he is not saying that he is coming,
			but he has not yet decided whether to say that he is coming,
			and it is potentially possible that he both goes to help and
			says that he is coming.

			The following model $C_3$ with just three points preserves all these distinctions.

			(The verification of that fact is too detailed to give here.)

			To indicate the simplest part of $C_3$ first: Let
			\begin{eqnarray*}
				W        &=& \{x, y, z\}\\
				\va(x)   &=& \{x, y\}\\
				\vp(x)   &=& W\\
				\norm{A} &=& \{z\}\\
				\norm{B} &=& \{y, z\}
			\end{eqnarray*}
			This is the only way to give the propositions $A$ and $B$ the correct modal status
			using only three points.

			We want the extended Chisholm set to be satisfied in the point $x$. Regarding
			$y$ and $z$ we may, for example, let
			$\va(y) = \vp(y) = \{y\}$ and
			$\va(z) = \vp(z) = \{z\}$.

			The hard part is to find a $\pii$ which is stabile under
			the four semantic conditions on $\pi$.

			After some computer computation we ended up with
			\begin{eqnarray*}
				\pii( \emptyset ) = \pii(\{y\})                         &=&  \emptyset\\
				\pii(\{z\}) = \pii(\{x, z\}) = \pii(\{y, z\}) = \pii(W) &=& U(\{z\})\\
				\pii(\{x\}) = \pii(\{x, y\})                    &=& U(\{x\})
			\end{eqnarray*}
			We see that $\norm{A \iff B}$ is the smallest context that
			belongs to each nonempty $\pii(X)$, which is reasonable.
			We also see that in the context where he does not go,
			but does say that he is coming,
			we can do without duties whatsoever, i.e., $\pii( \norm{\neg A \land B} ) = \emptyset$.

			The pragmatic oddity has been avoided in the sense that
			\[
				( \forall X)(  \norm{\neg  A  \land  B} =
				\{y\}  \in  \pii(X)).
			\]
			It is more tricky to design models for classical than for
			normal conditional logic.
			But it can be done. Creativity does not seem to be required;
			one starts with all relations empty and looks at each condition
			to see what needs to be added to the relations.
			Additionally one must check that the model does not become too rich.
			Two disjoint sets
			cannot both be in $\pii(X)$, for example, as this would contradict
			\[
				\tag{\ref{cj3} revisited}
				(Y  \in  \pii(X)  \land  Z  \in  \pii(X))  \rightarrow
				(Y  \cap  Z  \in  \pii(X))
			\]
			and
			\[
				\tag{\ref{cj1} revisited}  \emptyset \not\in  \pii(X).
			\]
			But again, it is easiest to let a computer do this labor.

			Tautologies are never obligatory (mandatory) in \CJ\
			in the sense that we would have $\bigcirc_a \top$ or $\bigcirc_i \top$.
			We may have $\bigcirc ( \top \mid A)$ as that just means that all nonempty contexts $X$
			contained in $A$ satisfy $W \in \pii(X)$. We see that
			$C_3 \models \bigcirc (\top\mid A)$.

		\subsection{Three interpretations of the dog scenario}
			Consider the following sentences.
			\[
				\tag{$D$} \text{There is a dog.}
			\]
			\[
				\tag{$S$} \text{There is a sign.}
			\]
			\[
				\tag{$F$} \text{There is a fence.}
			\]
			We now show how
			the expressive power of \CJ\ can be used in three different cases.
			The deontic component of the dog scenario is
			\[
				\bigcirc ( \neg  D\mid  \top)  \land  \bigcirc (S\mid D)  \land  \bigcirc (F\mid D  \land \neg  S)
			\]
			whereas the alethic components can vary:

			\paragraph{First case:} There is a dog, but not by factual necessity.

			\noindent Premises:
			\[
				D  \land \jonesDiamond{\rightarrow}   \neg  D
			\]
			Conclusion:
			\[
				\viol( \neg  D)  \land  \bigcirc_a  \neg  D
			\]
			\paragraph{Second case:} There is by factual necessity a dog.

			\noindent Premises:
			\[
				 \jonesBox{\rightarrow} D  \land \plainDiamond \neg  D
			\]
			Conclusion:
			\[
				\viol( \neg  D)  \land
				\jonesDiamond{\rightarrow} S  \land
				( \jonesDiamond{\rightarrow}   \neg  S  \rightarrow  \bigcirc_a S)
			\]

			\paragraph{Third case:}
			It is factually necessary that there is a dog and no sign,
			but it is potentially possible that there is a sign and no dog;
			there is no fence but it is factually possible that there is a fence.

			This is meaningful if we interpret factual necessity as the result of a
			decision, and potential possibility as coherence with practical or physical
			possibilities for action. This case constitutes the big test of strength for
			the system as far as whether it can represent notorious sinners --- agents
			who violate those duties that arise when (other) duties have been violated.

			\noindent Premises:
			\[
				      ( \jonesBox{\rightarrow} D)
				\land ( \plainDiamond \neg  D)
				\land ( \jonesBox{\rightarrow}   \neg  S)
				\land ( \plainDiamond  (D  \land  S))
				\land ( \neg  F)
				\land ( \jonesDiamond{\rightarrow} F)
			\]
			Conclusion:
			\[
				\viol( \neg  D)  \land
				\viol(D  \rightarrow  S)  \land
				\viol(D  \land \neg  S  \rightarrow  F)  \land
				\bigcirc_a F
			\]
			which is reasonable. This case incorporates CTCTD, violation of
			contrary-to-duty conditionals.
			The reason that we at all can deduce $\viol$ is Result 2, which
			allows us to deduce $\bigcirc (A \rightarrow B)$ given that $\bigcirc (B\mid A)$.

		\subsection{Calculation of models of \CJ\ using Maple}
			It is of interest to construct a model of the dog-sign-fence scenario,
			in order to judge how the system \CJ\ tackles 2nd Level CTDs.
			To consider this set of sentences abstracted from their meaning, let us relabel alphabetically:
			\[
				\tag{$A$} \text{There is a dog.}
			\]
			\[
				\tag{$B$} \text{There is a sign.}
			\]
			\[
				\tag{$C$} \text{There is a fence.}
			\]
			The set of propositions that we seek a model for is:
			\begin{eqnarray*}
				\bigcirc ( \neg  A\mid  \top)\\
				\bigcirc ( \neg  B\mid  \neg  A)\\
				\bigcirc (B\mid A)\\
				\bigcirc (C\mid A  \land \neg  B)\\
				 \jonesBox{\rightarrow} (A  \land \neg  B)\\
				 \plainDiamond \neg  A\\
				 \plainDiamond  A  \land  B\\
				 \neg  C\\
				 \jonesDiamond{\rightarrow} C
			\end{eqnarray*}
			Can we construct a model consisting of only four worlds? Let
			\[
				W = \{a, b, c, d\},\quad
				A = \{a, b, d\},\quad
				B = \{d\},\quad
				C = \{b\}.
			\]
			Let $\va(a) = \{a, b\}$ and $\vp(a) = W$. We would like the set of propositions
			to be satisfied in world $a$. \emph{Ipso facto}, the propositions where $\pii$ does not
			occur in the corresponding truth conditions are satisfied.

			Using \textsf{Maple} we input some values for $\pii$ obtained from
			the truth conditions of the four propositions that
			are of the form $\bigcirc (\cdot\mid\cdot)$ in our extended Chisholm set.
			The extension of this to a complete definition of $\pii$
			was best left to a computer.
			We found that
			$\pii(X) = \emptyset$ except for the case where
			$X$ = $\emptyset$ and $X = \{a\}$.
			This is quite satisfying since
			\[
				\{a\} = \norm{A \land  \neg  B  \land  \neg  C}.
			\]
			Given that we are stuck in a context where there is a dog, no sign, and no fence,
			we are in a hopeless situation which deontically speaking
			is best compared to the contradictory context $\emptyset$.

		\newpage
		\subsection{A new result}
			In \cite{CJ96} the authors mention that (\ref{cj5}),
			\[
				\tag{\ref{cj5} revisited}
				(Z \in \pii(X)) \land (Y \subseteq X) \land (Y \cap Z \ne \emptyset)
				\rightarrow (Z \in \pii(Y )),
			\]
			might be a reasonable semantic condition on $\pii$ to add to
			(\ref{cj1}), (\ref{cj2}), (\ref{cj3}) and (\ref{cj4}).

			\begin{lem}
				Suppose $\bigcirc (A\mid\top)$ is true in a model of \CJ\ where $W = \{a, b, c\}$ and $\norm{A} = \{a\}$.
				Then $\pii$ must satisfy
				\[
					( \forall X  \subseteq  W)(a  \in X \rightarrow  \pii(X) = U(\{a\})).
				\]
			\end{lem}
			\begin{proof}
				This can be reproduced by hand or using a computer, by going through
				(\ref{cj1}), (\ref{cj2}), (\ref{cj3}) and (\ref{cj4}) until $\pii$ has stabilized.
			\end{proof}

			\begin{thm}\label{mainThm}
				Suppose $\bigcirc (A\mid\top)$ is true in a model of \CJ\ where $W = \{a, b, c\}$ and $\norm{A} = \{a\}$.
				Then (\ref{cj5}) fails.
			\end{thm}
			\begin{proof}
				Suppose (\ref{cj5}) does hold. Let
				\begin{eqnarray*}
					Y   &=& \{b, c\},\\
					X   &=& W,\\
					Z_1 &=& \{a, b\},\\
					Z_2 &=& \{a, c\}.
				\end{eqnarray*}
				We deduce
				\begin{eqnarray*}
					Z_1 \in  \pii(Y )  \land  Z_2 \in  \pii(Y )         & \text{from condition } & (5)\\
					\{b\}  \in  \pii(Y )  \land  \{c\}  \in  \pii(Y ) & &(2)\\
					 \emptyset \in  \pii(Y )                          & &(3)\\
					 \bot                                            & &(1)\\
				\end{eqnarray*}
				and the Theorem has been proved.
			\end{proof}
			The semantic conditions we needed in order to get from each line to the next are thus
			(\ref{cj1}), (\ref{cj2}), (\ref{cj3}) and (\ref{cj5}). To prove the Lemma we also needed (\ref{cj4}).

			Thus, we cannot use (\ref{cj5}) together with the other semantic conditions on $\pii$
			to formalize deontic scenarios. The condition (\ref{cj5}) expresses a
			supposition that the mandatory in a context is preserved when passing to a more specific
			context, as long as the fulfillment of the mandate is compatible with the new context.

	\nocite{
		MR0444437,
		lindahl1977position,
		Elgesem,
		MR700544,
		hart1994concept,
		MR3063042}

	\newpage
	\bibliographystyle{plain}
	\bibliography{filhspe1}

	\newpage
	\section{Appendix: Maple}
	\includegraphics[width=14cm]{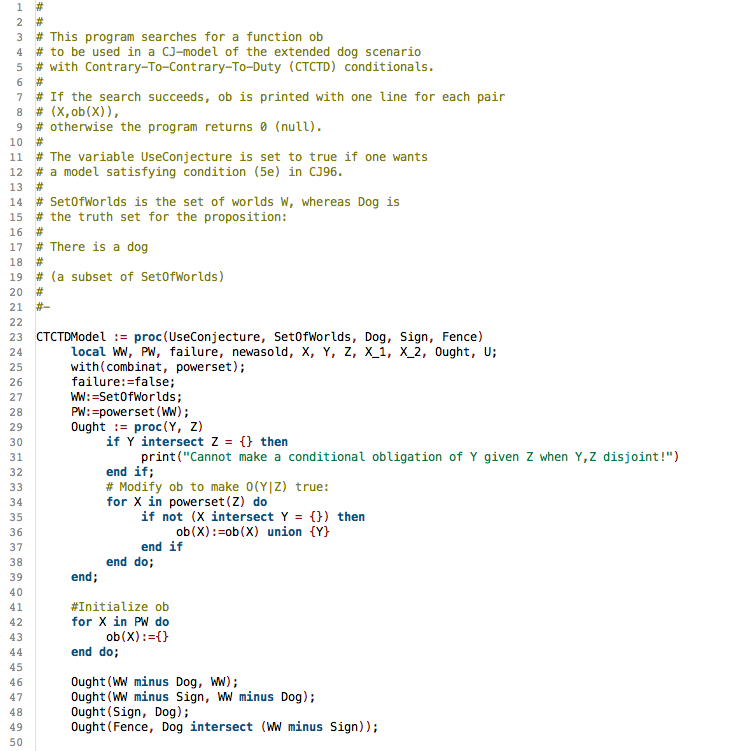}
	\includegraphics[width=14cm]{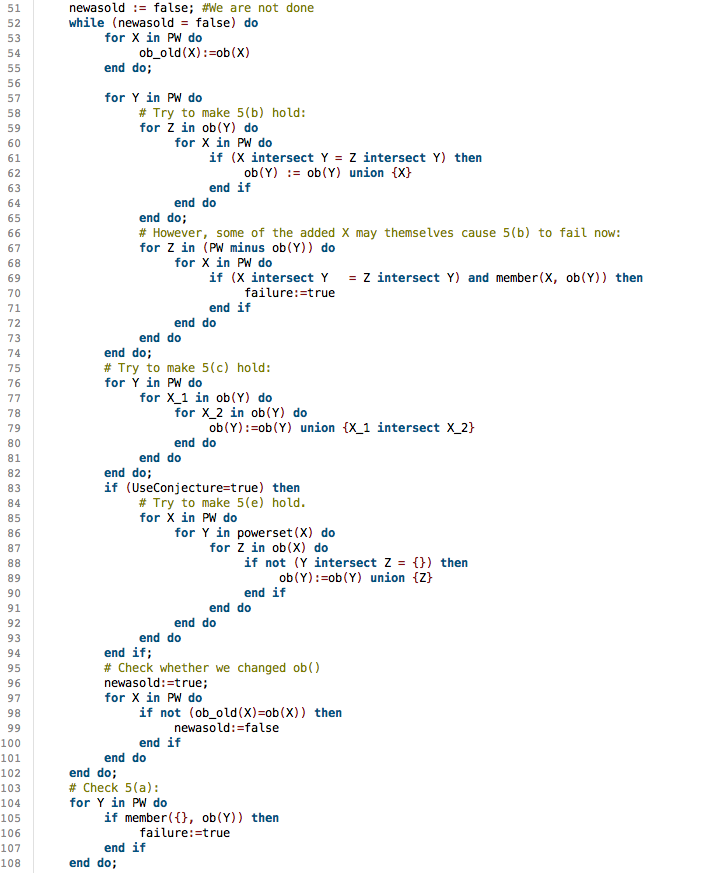}
	\includegraphics[width=14cm]{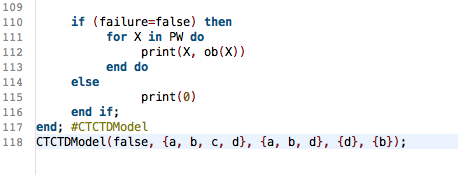}
	\newpage
	\subsection*{Maple code}
	\begin{verbatim}
		#
		#
		# This program searches for a function ob
		# to be used in a CJ-model of the extended dog scenario
		# with Contrary-To-Contrary-To-Duty (CTCTD) conditionals.
		#
		# If the search succeeds, ob is printed with one line for each pair
		# (X,ob(X)),
		# otherwise the program returns 0 (null).
		#
		# The variable UseConjecture is set to true if one wants
		# a model satisfying condition (5e) in CJ96.
		#
		# SetOfWorlds is the set of worlds W, whereas Dog is
		# the truth set for the proposition:
		#
		# There is a dog
		#
		# (a subset of SetOfWorlds)
		#
		#-

		CTCTDModel := proc(UseConjecture, SetOfWorlds, Dog, Sign, Fence)
			local WW, PW, failure, newasold, X, Y, Z, X_1, X_2, Ought, U;
			with(combinat, powerset);
			failure:=false;
			WW:=SetOfWorlds;
			PW:=powerset(WW);
			Ought := proc(Y, Z)
				if Y intersect Z = {} then
					print("Cannot make a conditional obligation of Y given Z when Y,Z disjoint!")
				end if;
				# Modify ob to make O(Y|Z) true:
				for X in powerset(Z) do
					if not (X intersect Y = {}) then
						ob(X):=ob(X) union {Y}
					end if
				end do;
			end;

			#Initialize ob
			for X in PW do
				ob(X):={}
			end do;

			Ought(WW minus Dog, WW);
			Ought(WW minus Sign, WW minus Dog);
			Ought(Sign, Dog);
			Ought(Fence, Dog intersect (WW minus Sign));

			newasold := false; #We are not done
			while (newasold = false) do
				for X in PW do
					ob_old(X):=ob(X)
				end do;

				for Y in PW do
					# Try to make 5(b) hold:
					for Z in ob(Y) do
						for X in PW do
							if (X intersect Y = Z intersect Y) then
								ob(Y) := ob(Y) union {X}
							end if
						end do
					end do;
					# However, some of the added X may themselves cause 5(b) to fail now:
					for Z in (PW minus ob(Y)) do
						for X in PW do
							if (X intersect Y	= Z intersect Y) and member(X, ob(Y)) then
								failure:=true
							end if
						end do
					end do
				end do;
				# Try to make 5(c) hold:
				for Y in PW do
					for X_1 in ob(Y) do
						for X_2 in ob(Y) do
							ob(Y):=ob(Y) union {X_1 intersect X_2}
						end do
					end do
				end do;
				if (UseConjecture=true) then
					# Try to make 5(e) hold.
					for X in PW do
						for Y in powerset(X) do
							for Z in ob(X) do
								if not (Y intersect Z = {}) then
									ob(Y):=ob(Y) union {Z}
								end if
							end do
						end do
					end do
				end if;
				# Check whether we changed ob()
				newasold:=true;
				for X in PW do
					if not (ob_old(X)=ob(X)) then
						newasold:=false
					end if
				end do
			end do;
			# Check 5(a):
			for Y in PW do
				if member({}, ob(Y)) then
					failure:=true
				end if
			end do;

			if (failure=false) then
				for X in PW do
					print(X, ob(X))
				end do
			else
					print(0)
			end if;
		end; #CTCTDModel
		CTCTDModel(false, {a, b, c, d}, {a, b, d}, {d}, {b});
\end{verbatim}
\newpage
\subsection*{Output}
\begin{verbatim}
{}, {}
{a}, {}
{b}, {{b}, {a, b}, {b, c}, {b, d}, {a, b, c}, {a, b, d}, 
{b, c, d}, {a, b, c, d}}
{c}, {{c}, {a, c}, {b, c}, {c, d}, {a, b, c}, {a, c, d}, 
{b, c, d}, {a, b, c, d}}
{d}, {{d}, {a, d}, {b, d}, {c, d}, {a, b, d}, {a, c, d}, 
{b, c, d}, {a, b, c, d}}
{a, b}, {{b}, {b, c}, {b, d}, {b, c, d}}
{a, c}, {{c}, {b, c}, {c, d}, {b, c, d}}
{a, d}, {{d}, {b, d}, {c, d}, {b, c, d}}
{b, c}, {{c}, {a, c}, {c, d}, {a, c, d}}
{b, d}, {{d}, {a, d}, {c, d}, {a, c, d}}
{c, d}, {{c}, {a, c}, {b, c}, {a, b, c}}
{a, b, c}, {{c}, {c, d}}
{a, b, d}, {{d}, {c, d}}
{a, c, d}, {{c}, {b, c}}
{b, c, d}, {{c}, {a, c}}
{a, b, c, d}, {{c}}

	\end{verbatim}

\end{document}